\documentclass[
  runningheads
]{llncs}

\usepackage[T1]{
  fontenc
}

\usepackage[
  scaled = 0.76
] {beramono}

\usepackage{
  adjustbox,
  alphalph,
  amsfonts,
  amsmath,
  amssymb,
  bm,
  graphicx,
  listings,
  mathtools,
  nicefrac,
  setspace,
  tikz,
  xcolor
}

\usepackage{subcaption}
\usepackage{mleftright}
\usepackage{todonotes}
\usepackage[
  framemethod = tikZ
]{mdframed}

\usepackage [inline] {
  enumitem
}

\usepackage {hyperref}
\hypersetup {
   colorlinks = true,
   citecolor = blue,
   urlcolor = blue,
   linkcolor = blue
}

\usepackage{cleveref}

\crefname{example}{Example}{examples}
\Crefname{example}{Example}{Examples}
\crefname{figure}{Fig.}{figures}
\Crefname{figure}{Figure}{Figures}
\crefname{section}{Sect.}{sections}
\Crefname{section}{Section}{Sections}
\crefname{table}{Tbl.}{tables}
\Crefname{table}{Table}{tables}
\crefname{theorem}{Thm.}{theorems}
\Crefname{theorem}{Theorem}{Theorems}

\usepackage{wrapfig}
\usepackage{multirow}
\usepackage{booktabs}
\definecolor{deepblue}{rgb}{0,0,0.5}
\definecolor{deepred}{rgb}{0.6,0,0}
\definecolor{deepgreen}{rgb}{0,0.5,0}
\usepackage{mathpartir}

\usepackage{amsthm}

\newcommand{\ncondrule}[3]{
  \begin{array}{c}
    \textsc{ ({#1})} \\[1pt]
    #2 \vspace*{1mm}\\[2pt]
    \hline \vspace{0.5ex}\\[-8pt]
    #3
  \end{array} }

\newcommand{\s}{\ensuremath{\mathcal{S}} }
\newcommand{\smean}{{\ensuremath{\vect{\mu}}} }
\newcommand{\scov}[1]{\ensuremath{\matr{\Sigma}^{#1}} }
\newcommand{\scovprime}{\ensuremath{\matr{\Sigma}'} }
\newcommand{\scovdoubleprime}{\ensuremath{\matr{\Sigma}''} }
\newcommand{\scovtripleprime}{\ensuremath{\matr{\Sigma}'''} }

\newcommand{\sstore}{\ensuremath{\sigma} }

\newcommand \COV {\textnormal {\textrm {cov}}}

\newcommand{\lift}{\ensuremath{\mathit{lift}}}
\newcommand{\program}{\ensuremath{\pi}}
\newcommand{\supp}{\ensuremath{\textit{supp}}}

\newcommand{\rvinput}{\ensuremath{\mathit{I}}}
\newcommand{\rvoutput}{\ensuremath{\mathit{O}}}
\newcommand{\rvage}{\ensuremath{\mathit{A}}}

\newcommand{\expr}{\ensuremath{e}}
\newcommand{\dist}{\ensuremath{d}}
\newcommand{\stmt}{\ensuremath{s}}
\newcommand{\pgrm}{\ensuremath{p}}

\newcommand{\Reals}{\ensuremath{\mathbb{R}}}
\newcommand{\Inputs}{\ensuremath{\mathbb{I}}}
\newcommand{\Outputs}{\ensuremath{\mathbb{O}}}
\newcommand{\Dists}{\ensuremath{\mathcal{D}}}
\newcommand{\Evi}{\ensuremath{\mathit{E}}}

\renewcommand \vec [1]
  {\bm {#1}}
\newcommand \vect [1]
  {\ensuremath {\vec {#1}}}

\newcommand{\matr}[1]{\ensuremath{\mathrm{#1}}}

\newcommand{\std}{\sigma^{2}} 

\newcommand{\gauss}[2] {\mathcal{N}(#1 , #2)}

\newcommand \covgauss
  { \ensuremath {
    ( (2\pi)^n | \matr \Sigma | ) ^{-\nicefrac12}
     \exp{\left(-{2^{-1}}(\vec x -  \vect{\mu} )^\intercal \matr{\Sigma}^{-1} (\vec x -  \vect{\mu} ) \right)} }}

\newcommand{\uniform}[1]{\mathcal{U}(#1)}

\newcommand{\ruleexpop}{\ensuremath{\text{\sc O-Exp}}}
\newcommand{\ruleexpvar}{\ensuremath{\text{\sc V-Exp}}}
\newcommand{\ruleexpconst}{\ensuremath{\text{\sc C-Exp}}}
\newcommand{\ruledetassig}{\ensuremath{\text{\sc D-Asg}}}
\newcommand{\ruleseq}{\ensuremath{\text{\sc Seq}}}
\newcommand{\ruleforb}{\ensuremath{\text{\sc For-B}}}
\newcommand{\rulefori}{\ensuremath{\text{\sc For-I}}}
\newcommand{\rulepassigind}{\ensuremath{\text{\sc P-Asg-Ind}}}
\newcommand{\rulepassigdep}{\ensuremath{\text{\sc P-Asg-Dep}}}
\newcommand{\ruleopplusminus}{\ensuremath{\text{\sc P-Op-PM}}}
\newcommand{\ruleopmuldiv}{\ensuremath{\text{\sc P-Op-MD}}}
\newcommand{\ruleoprvs}{\ensuremath{\text{\sc P-Sum}}}
\newcommand{\ruleopcond}{\ensuremath{\text{\sc P-Cond}}}
\newcommand{\rulereturn}{\ensuremath{\text{\sc Ret}}}

\begin{document}
%

\title{Exact and Efficient Bayesian Inference for Privacy Risk Quantification\thanks{Work partially supported by funding from the topic Engineering Secure Systems of the Helmholtz Association (HGF), the KASTEL Security Research Labs and the Danish Villum Foundation through Villum Experiment project No. 0002302.} \\ {\large{(Extended Version)}}}
\titlerunning{Exact and Efficient Bayesian Inference for Privacy Risk Quantification}
%
\author{Rasmus C. R{\o}nneberg \inst{1}  \and
Ra\'ul Pardo \inst{2} \and
Andrzej W\k{a}sowski \inst{2}}
\authorrunning{R. C. R{\o}nneberg et al.}
%
\institute{ Karlsruhe Institute of Technology, Germany \\ \email{rasmus.ronneberg@kit.edu}
\and
IT University of Copenhagen, Denmark \\
\email{\{raup,wasowski\}@itu.dk}
}
\maketitle              
\begin{abstract}
  Data analysis has high value both for commercial and research purposes.
  However, disclosing analysis results may pose severe privacy risk to individuals.
  Privug is a method to quantify privacy risks of data analytics programs by analyzing their source code.
  The method uses probability distributions to model attacker knowledge and Bayesian inference to update said knowledge based on observable outputs.
  Currently, Privug uses Markov Chain Monte Carlo (MCMC) to perform inference, which is a flexible but approximate solution.
  This paper presents an exact Bayesian inference engine based on multivariate Gaussian distributions to accurately and efficiently quantify privacy risks.
  The inference engine is implemented for a subset of Python programs that can be modeled as multivariate Gaussian models.
  We evaluate the method by analyzing privacy risks in programs to release public statistics.
  The evaluation shows that our method accurately and efficiently analyzes privacy risks, and outperforms existing methods.
  Furthermore, we demonstrate the use of our engine to analyze the effect of differential privacy in public statistics.
  \looseness -1
\end{abstract}
\section{Introduction}
%
%
%

Data anonymization methods gain legal importance\,\cite{article_29_data_protection_working_party_opinion_nodate} as data collection and analysis are expanding dramatically in data management and statistical research.  Yet applying anonymization, or understanding how well a given analytics program hides sensitive information, is non-trivial\,\cite{elliot_anonymisation_2016}. Contemporary anonymization algorithms, such as  differential privacy\,\cite{dp}, require calibration to balance between reducing risks and preserving the utility of data.  To assess the risks, data scientists need to assess the flow (leakage) of information from sensitive data fields to the output of analytics.

Measuring the information leakage is a useful technique to quantify how much an attacker may learn about the sensitive information a program processes.
Many methods have been proposed in this domain~\cite{qifbook.2020,chothia.leakest.2013,chothia.leakwatch.2014,QUAIL,HyLeak,spire,cherubin.fbleau.2019,romanelli.leaves.2020,privug}.
Privug is a recent one\,\cite{privug}. It relies on Bayesian inference to quantify privacy risks in data analytics programs.
The attacker's knowledge is modeled as a probability distribution over program inputs, and it is then conditioned on the disclosed program outputs.
Then, Bayesian probabilistic programming is used to compute the posterior attacker knowledge, i.e., the updated attacker knowledge after observing the outputs in the program.
One of the advantages of Privug is that it works on the program source code, and can be extended to compute most information leakage metrics~\cite{qifbook.2020}.
However, Privug currently relies on approximate Bayesian inference such as Markov Chain Monte Carlo~\cite{mcmc}.
These methods can be used to analyze arbitrary programs, but may be computationally expensive and may produce imprecise results.
In quantifying privacy risks, precision is critical, as under-estimation of risks may result in an illegal disclosure of personal information.
\looseness -1

We present a new \emph{exact} and \emph{efficient} Bayesian inference engine for Privug targeting attackers modeled by multivariate Gaussian distributions. Even though not all standard program statements can be mapped to operations on Gaussian distributions (and thus not all programs can be analyzed using our new engine), multivariate Gaussian distributions are a good candidate for a semantic domain of attacker's knowledge for several reasons:
\begin{enumerate*}[label=\roman*)]

  \item Multivariate Gaussians are closed under many common operations and they can be computed efficiently;
  
  \item The Gaussian distribution is a \textit{maximum entropy} distribution under common conditions~\cite{mcelreath2020statistical,jaynes} that allows modeling prior attacker knowledge with minimum assumptions;

  \item Gaussians are commonly used in probabilistic modeling of engineering systems to represent uncertainty of measurement.

\end{enumerate*}

This work constitutes a new point in the study of expressiveness vs performance in quantification of privacy risks by means of Bayesian inference.
Specifically, our contributions are:%
\begin{enumerate}

  \item A probabilistic programming language for exact Bayesian inference using multivariate Gaussians. The language is a subset of Python (\cref{sec:Exact}).

  \item A definition of a sound (\cref{sec:properties}) Bayesian inference engine (\cref{sec:semantics}) using multivariate Gaussian distributions.

  \item A proof-of-concept implementation of the inference engine as a library that can be applied to analyze our subset of Python.

  \item An application of the inference engine to a case study of privacy risk quantification in public statistics (\cref{sec:case-study}), with and without differential privacy\,\cite{dp}.
    %

  \item An evaluation of the scalability of the inference engine, and a comparison with existing inference methods for privacy risk quantification (\cref{sec:eva:scalability}).
    The evaluation shows that our engine can analyze large systems involving thousands of individuals, and also that we greatly outperform existing tools for the programs supported in our language.
\end{enumerate}

\noindent The code to reproduce the evaluation and case study in this paper is available at~\cite{gauss-privug-doi}. The proof-of-concept implementation of the exact inference engine is an open source project available at \href{https://github.com/itu-square/gauss-privug}{https://github.com/itu-square/gauss-privug}.


\section{Background}\label{sec:preliminaries}

\subsection{Privug: A Data Privacy Debugging Method}
\label{sec:privug}

Let $\Inputs, \Outputs$ denote sets of \emph{inputs} and \emph{outputs}, respectively.
We use $\Dists(\Inputs)$ to denote a space of distributions; in this case over inputs.
Let $d \in \Dists(\Inputs)$ denote a distribution over inputs, $\rvinput \sim \Dists(\Inputs)$ denotes a random variable distributed according to $d$.

Privug~\cite{privug} is a method to explore information leakage on data analytics programs.
The method combines a probabilistic model of attacker knowledge with the program under analysis to quantify privacy risks.
This process is summarized in the following steps:

\paragraph{(1) Prior.}
We first model the \emph{prior} knowledge of an attacker as a distribution over program inputs.
This distribution represents the input values of a program that the attacker finds plausible.
For example, consider a program that takes as input a real number \(\rvage\) representing the age of an individual (\(\Inputs \triangleq \Reals\)).
A possible prior knowledge of the attacker could be: $\rvage \sim \uniform{0, 120}$ (all ages between $0$ and $120$ are equally likely for the attacker) or $\rvage \sim \gauss{\mu=42}{\sigma=2}$ (the attacker believes that the age of 42 and values nearby are most likely ages).
We write \(P(\rvinput)\) for the distribution of prior attacker knowledge.

\paragraph{(2) Probabilistic program interpretation.}

The second step is to interpret a target program $\program : \Inputs \to \Outputs$ using the attacker prior knowledge.
To this end, we lift the program to run on distributions $\Dists(\Inputs)$ instead of concrete inputs $\Inputs$.
This corresponds to the standard lifting to the probability monad~\cite{GordonHNR14}; $\lift : (\Inputs \to \Outputs) \to (\Dists(\Inputs) \to \Dists(\Outputs))$.
For example, consider the following program that computes the average age of a list of ages (in Python):
\begin{lstlisting}
def average_age(ages: List[float]): return sum(ages)/len(ages)
\end{lstlisting}
The lifted version of the program is (Python allows retaining the same body):
\begin{lstlisting}
def average_age(ages: Dist[List[float]] ): return sum(ages)/len(ages)
\end{lstlisting}
where \lstinline|Dist[List[float]]| denotes a distribution over lists of floats, $\Dists(\Reals^n)$.
The lifted program yields the distribution $P(\rvoutput|\rvage)$.
In general, the combination of prior attacker knowledge with the lifted program yields a joint distribution on inputs and outputs, $P(\rvoutput|\rvinput)P(\rvinput) = P(\rvoutput,\rvinput)$.

\paragraph{(3) Observations.}
It is possible to analyze privacy risks for concrete outputs of the program.
To this end, one may add observations to the probabilistic model.
In the average example above, we could check how the knowledge of the attacker changes when the attacker observes that the average is $44$.
This step yields the posterior distribution $P(\rvinput | \rvoutput = 44)$.
In general, the probabilistic lifting of the program \program\ defines a likelihood function on the joint distribution of input and output, $P(\Evi | \rvinput,\rvoutput)$ with  some predicate $\Evi$ on the joint distribution.

\paragraph{(4) Posterior Inference.}
The next step is to apply Bayesian inference to obtain a posterior distribution on the input variables.
\begin{equation}
  P(\rvinput,\rvoutput \vert \Evi) = {P(\Evi \vert \rvinput,\rvoutput)  P(\rvoutput \vert \rvinput)  P(\rvinput)} / {P(\Evi)}
\end{equation}
It is usually intractable to compute a symbolic representation
of $P(\Evi)$.
Therefore, it is not possible to get analytical solutions for the posterior distributions.
The current implementation of Privug uses Markov Chain Monte Carlo (MCMC)\,\cite{mcmc} to tackle this issue.
But MCMC methods are approximate and do not always converge.
As mentioned above, the subject of this paper is to provide an exact inference method to efficiently and precisely compute the posterior distribution.

\paragraph{(5) Posterior analysis}
We query the posterior and prior distributions (attacker knowledge) to measure how much the attacker has learned.
This can be done by applying different techniques such as computing probability queries, plotting visualizations of probability distributions, computing information theoretic metrics (e.g., entropy or mutual information), and metrics from \textit{quantitative information flow}~\cite{qifbook.2020} such as Bayes vulnerability.
\looseness-1

\begin{wrapfigure}{r}{.4\textwidth}
  \vspace{-8mm}
  \centering
  \includegraphics[width=.4\textwidth]{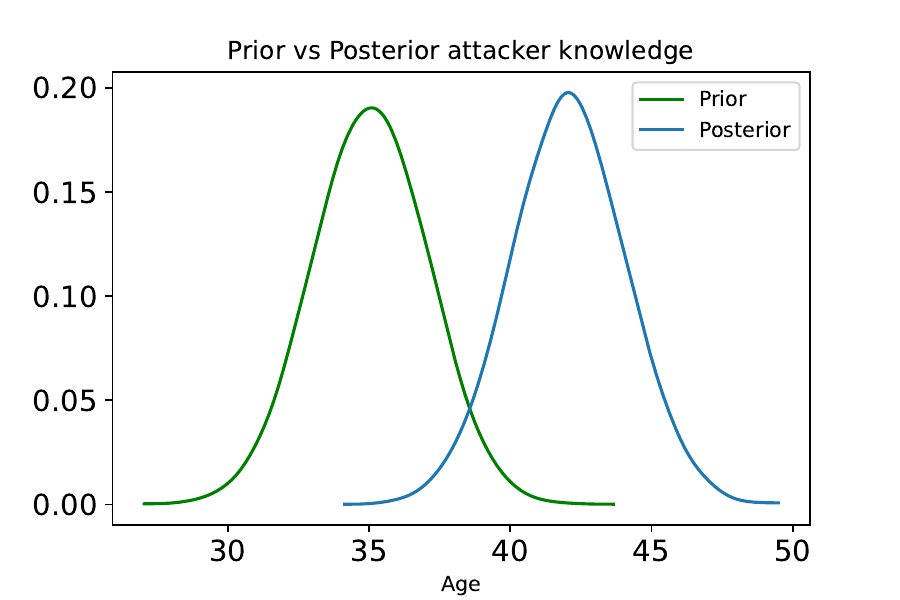}
  \caption{Prior vs Posterior ages}%
  \label{fig:privug_example}
  \vspace{-3mm}
\end{wrapfigure}

For example, \cref{fig:privug_example} compares the prior and posterior distributions of the average age program above.
We analyze the case where the output of the program is $44$.
The green line shows the prior attacker knowledge on the victim's age $P(\rvage)$ and the blue line the posterior knowledge $P(\rvage|\rvoutput=44)$ when observing that the program output is $44$.
The prior attacker knowledge is $\rvage \sim \gauss{35}{2}$ for the victim's age and also for the rest.
The figure clearly shows that the attacker now believes that higher ages are more plausible.
In other words, the attacker prior knowledge has been corrected towards more accurate knowledge on the victim's age.

\subsection{Multivariate Gaussian Distributions}%
\label{sec:ggm}
\label{sec:gaussian-background}

In this paper, we use capital Greek letters for matrices, and bold font for column vectors. Small letters \(a\) and \(b\) are reserved for selecting subvectors (as in \(\vec \mu_a\)) and pairs of them for selecting submatrices (as in \(\matr \Sigma_{ba}\)). Matrix and vector literals are written in brackets. We write  \(\supp(X)\) for the support of the random variable \(X\).
\looseness -1

\looseness -1
A \emph {multivariate Gaussian distribution}, denoted \( \vect X \sim \gauss { \vect \mu } {\matr \Sigma} \), defines a probabilistic model composed of \( n \) normally distributed random variables, \( \vec {X} = [ X_1, X_2, \ldots, X_n ] ^\intercal \). The distribution is parameterized by a vector \( \vec {\mu} \) of \( n \) means, and a symmetric \( n \times n \) \emph{covariance} matrix \( \matr \Sigma \), so \( \matr \Sigma_{ij} = \COV [X_i, X_j] \) gives the covariance between \( X_i \) and \( X_j \), \( \matr \Sigma_{kk} \) gives the variance of \( X_k \). We assume that the covariance matrix is positive definite. The probability density function is:
\begin {equation}
  P ( \vec x ) = \covgauss
  \label {eq:multigauss-pdf}
\end {equation}
where \( | \matr \Sigma | \) denotes the determinant of the matrix \( \matr \Sigma \). We recall standard properties of multivariate Gaussian distributions \cite{eaton2007multivariate,PGM2009,bishop}.
\begin{theorem}\label{thm:marginals}
  Let 
  \( 
  \begin{bmatrix}
    \vec X_a \\ \vec X_b
  \end{bmatrix}
  \sim \gauss {\vec \mu} {\matr \Sigma} 
  \) 
  with
  \( 
  \vec \mu = 
  \begin{bmatrix}
    \vec \mu_a \\ \vec \mu_b
  \end{bmatrix}
  \)
  and
  \(
    \matr \Sigma = \begin {bmatrix}
                    \matr \Sigma_{aa} & \matr \Sigma_{ab} \\
                    \matr \Sigma_{ba} & \matr \Sigma_{bb}
                  \end {bmatrix}
  \).
  The marginal distributions are
  \( \vec X_a \sim \gauss { \vec \mu_a } { \matr \Sigma_{aa} } \),
  \( \vec X_b \sim \gauss { \vec \mu_b } { \matr \Sigma_{bb} } \), and
  \( \vec X_i \sim \gauss { \vec \mu_i } { \matr \Sigma_{ii} } \)
  for \( i = 1 \dots a+b \).
\end{theorem}

\noindent \looseness -1
The covariance matrix identifies independent random variables:
\begin {theorem}\label {thm:independence}
Let \( [ X_1, \dots, X_n ] ^\intercal \sim \gauss {\vec \mu } { \matr \Sigma } \), two marginals \(X_i\), \(X_j\) with \( i \not= j \) are \emph{independent} iff \( \matr \Sigma_{ij} = \COV [X_i, X_j] = 0 \).
\end {theorem}

\noindent \looseness -1
The space of Gaussian distributions is closed under affine transformations:
\begin {theorem} \label {thm:affine}
  Let \( \vec X \sim \gauss { \vec \mu } { \matr \Sigma } \) and \( \vec Y = \matr A \vec X + \vec b \) be an affine transformation with \( \matr A \in \Reals^{m \times n} \) a projection matrix and \( \vec b \in \Reals^{n \times 1} \) a column vector.
  Then, \( \vec Y \sim \gauss { \matr A \vec \mu + \vec b } { \matr A \matr \Sigma \matr A^\intercal } \) holds.
\end{theorem}

\noindent
We use \( Y | X_1, X_2, \dots, X_n \) to denote a random variable \( Y \) that is distributed conditionally with respect to \( X_1, X_2, \ldots, X_n \). Linear combinations of random variables can be used to define hierarchical probabilistic models consisting of dependent random variables, such as Gaussian Bayesian networks~\cite{PGM2009}.
\begin {theorem} \label {thm:linear}
  Let \( \vec X \sim \gauss { \vec \mu } { \matr \Sigma } \) and \( Y | \vec X \sim \gauss { \vec {a}^ \intercal \vec {X} + b } { \sigma^2 } \), where \( \vec a \in \Reals^{n \times 1} \) is a vector, \( b \in \Reals \) and \(\sigma^2 > 0\). Then \( [ \vec X ^ \intercal, Y ] ^ \intercal \sim \gauss { [\smean ^ \intercal, \vec a ^ \intercal \vec \mu + b ] ^ \intercal } { \scovprime } \) with
  \begin {equation*}
    \matr \Sigma' _{1..n, 1..n} \! \! = \! \matr \Sigma,
    ~
    \matr \Sigma' _{(n+1)(n+1)} = \sigma^2 + \vec a^\intercal \scov{} \vec a ,
    ~
    \matr \Sigma'_{i(n+1)} = \COV [X_i, Y] = \textstyle \sum_{j=1}^n a_j\matr \Sigma _{ij} \enspace .
  \end {equation*}
\end {theorem}

\smallskip

\begin {example}
\label {ex:gaussian-bayesian-network-simple-chain}
We present an example of a Gaussian Bayesian network\,\cite {PGM2009}. Let \( X_1 \sim \gauss {50} {2} \), \( X_2 | X_1 \sim \gauss {2X_1 - 5} {1} \),  and \( X_3 | X_2 \sim \gauss {X_2 - 10} {4} \). Here  the distribution of \( X_2 \) is conditioned on \( X_1 \), and of \( X_3 \) on \( X_2 \). The model defines a joint multivariate Gaussian probability distribution \( [ X_1, X_2, X_3 ]^ \intercal \sim \gauss {\vec \mu} {\matr \Sigma} \). \Cref{thm:linear} allows to compute the mean, variance, and covariance of this joint distribution:
\looseness -1

\vspace{-2mm plus .5mm minus 1mm}

\begin{equation*}\footnotesize
  \vec \mu \! = \!
  \begin{bmatrix}
   50\\
   2 \cdot \vec \mu_1 - 5\phantom0 \\
   1 \cdot \vec \mu_2 - 10
   \end{bmatrix}
   \! = \!
   \begin{bmatrix}50 \\ 95 \\ 85\end{bmatrix} \! ,
  \quad
  \matr \Sigma \! = \!
  \begin{bmatrix}
      ~ 2 ~
      &
      \phantom{1 +} 2 \cdot \matr \Sigma_{11} \phantom{\cdot 2}
      &
      ~ 0 + 1 \cdot \matr \Sigma_{12} \phantom{\, \, \cdot 1}
    \\      
      & 1 + 2 \cdot \matr \Sigma_{11} \cdot 2
      & ~ 0 + 1 \cdot \matr \Sigma_{22} \phantom{\, \, \cdot 2}
    \\
      ~
      &
      & ~ 4 + 1 \cdot \matr \Sigma_{22} \cdot 1
  \end{bmatrix}
  \! = \!
  \begin{bmatrix}
     ~ 2 & 4  & 4   \\
         & 9 &  9  \\
         &   & 13
  \end{bmatrix}
\end{equation*}

\vspace{0mm plus 0.5mm}

\noindent
As the matrices are symmetric, we only show the upper-right part. Note that, even tough \( X_3 \) does not directly depend on \( X_1 \), it still has a non-zero covariance.
The reason for this is the indirect dependence through $X_{2}$.\qed
\end{example}

\medskip

\noindent
We use \emph{conditioning} to model observations on values of random variables.
\begin {theorem}\label {thm:conditions}
  Let \( \vec X \sim \gauss {\vec \mu} { \matr \Sigma } \) be split into two sub-vectors so that
  \begin {equation*}
  \vec X = \begin {bmatrix}
               \vec X_a \\
               \vec X_b
             \end {bmatrix},
  \enspace
  \vec \mu = \begin {bmatrix}
                 \vec \mu_a\\
                 \vec \mu_b
               \end {bmatrix},
  \enspace
  \matr \Sigma = \begin {bmatrix}
                    \matr \Sigma_{aa} & \matr \Sigma_{ab} \\
                    \matr \Sigma_{ba} & \matr \Sigma_{bb}
                  \end {bmatrix}
  \enspace
  \text {and } \vec x_b \in \supp(\vec {X}_b).
  \end {equation*}
  The conditioned distribution is \( \vec X_a | (\vec X_b = \vec {x}_b) \sim \gauss{\vec {\mu}'}{\matr{\Sigma}'}\) with \(\vec {\mu}' = \vec {\mu}_{a} + \matr{\Sigma}_{ab}\matr{\Sigma}_{bb}^-(\vec {x}_b-\vec {\mu}_{b}) \) and \( \matr \Sigma' = \matr \Sigma_{aa} - \matr \Sigma_{ab}\matr \Sigma_{bb}^- \matr \Sigma_{ba} \), where $\matr{\Sigma}^-$ is the generalized inverse.
\end{theorem}

\begin{example}
Consider the multivariate distribution of~\cref{ex:gaussian-bayesian-network-simple-chain}. We condition \( X_3 \) to be \( 85 \). By \cref{thm:conditions} the posterior of \( X_1, X_2 | X_3 = 85 \) is \( \gauss{\vec {\mu}'}{\matr \Sigma'} \) with
\setlength\arraycolsep{2.1pt}
\begin{equation*}\footnotesize
  \vec \mu' \! = \!
    \begin{bmatrix} 50 \\ 95 \end{bmatrix} \! + \!
    \begin{bmatrix} 4 \\ 9 \end{bmatrix} \begin{bmatrix} 13 \end{bmatrix}^- \!
    (85 - 85) \! = \! \begin{bmatrix} 50 \\ 95 \end{bmatrix}
  \text {and }
  \matr \Sigma' \! = \!
    \begin{bmatrix} 2 & 4 \\ 4 & 9 \end{bmatrix} - \begin{bmatrix} 4 \\ 9 \end{bmatrix} \!
    \begin{bmatrix} 13 \end{bmatrix}^- \!
    \begin{bmatrix} 4 & \! 9 \end{bmatrix} \! = \!
    \begin{bmatrix} \nicefrac {10}{13} & \nicefrac {16}{13} \\ \nicefrac {16}{13} & \nicefrac {36} {13} \end{bmatrix} \! \! \!
    \raisebox{-2ex} {\qed}
\end{equation*}
\end{example}


\section{Exact Inference Engine for Privug}
\label{sec:Exact}

\looseness -1
Our inference engine is an interpreter of a probabilistic programming language that corresponds to a subset of Python. We include variable assignments, bounded for-loops, binary operators, sequencing, \textit{probabilistic assignments}, and \textit{observations} (conditioning). Let \( v_r \in \Reals\) be real values, \( x, y, z, \ldots \) denote (deterministic) variables, $X, Y, Z, \ldots$ be (Gaussian) random variables, and \( \vec X \) a vector of random variables. Let \( \oplus \in \{+,-,*,/\} \). The syntax of well-formed programs is generated by the rule \(\pgrm\) below.
\begin{equation*}
  \begin{array}{r@{\hskip 2.5ex} r@{\hskip 1.5ex} c@{\hskip 1.5ex} l}
    \text{(Expressions)}&\expr    &::=&   v_r \mid x \mid \expr \oplus \expr \\
    \text{(Distributions)}&\dist    &::=&   \texttt{\upshape Normal}(\expr,\, \expr) \mid \texttt{\upshape Normal}(\expr * X + \expr,\, \expr) \\
    \text{(Statements)}&\stmt    &::=&   X = \dist \mid X = Y \oplus \expr \mid X = Y + Z  \mid \texttt{\upshape condition}(X,\expr) \mid \\
    &~        & ~ &   x = \expr \mid \stmt\texttt{\upshape ; } \stmt \mid \texttt{\upshape for} \; x \; \texttt{\upshape in\;range} \; v_r \; \stmt \\
    \text{(Programs)}&\pgrm    &::=&   \stmt \texttt{\upshape; return} \; \vec {X} \qedhere
  \end{array}
\end{equation*}

\noindent
We admit (\expr) constants expressions, references to deterministic program variables, and binary operations. Two ways of defining normal distributions (\dist) are supported: an independent Gaussian distribution, or a linear transformation of random variables. Statements (\stmt) are:  probabilistic assignments (a normal, a transformed distribution, a sum of two random variables), an observation (conditioning), deterministic assignment, sequencing, and a limited for-loop.  We define no expressions over random variables, only statements, to simplify introduction of changes to the state (the probabilistic model) in the semantics for each sub-expression (\cref{sec:semantics}). The for-loops are only a convenience construct for repetitive statements. A program (\pgrm) terminates returning  a random variable (\texttt{return}). The distribution of the returned variable is the marginal of the posterior joint probability distribution that we want to reason about.
\looseness -1

Although the language appears restrictive, we show in~\cref{sec:case-study} that it can be used in realistic case studies; e.g., for the study of privacy risks in database reconstruction attacks using public statistics. Furthermore, this syntax ensures soundness and termination (\cref{sec:properties}) of a highly scalable (\cref{sec:eva:scalability}) inference engine.
\looseness -1

\subsection{Semantics}
\label{sec:semantics}
The formal semantics is defined in the small-step style, over terms of multivariate Gaussian distributions (\cref{sec:gaussian-background}).
It provides a sound and efficient inference engine to track attacker knowledge in Privug~(cf.~\cref{sec:privug}).

A \emph{state} \s is a tuple $\langle \smean, \scov{}, \sstore \rangle$.
The first two elements define a multivariate Gaussian distribution \( \gauss {\smean} {\scov{}} \) over $n$ random variables.
Let $\mathcal{V}$ denote the set of deterministic variables, $\sstore : \mathcal{V} \to \Reals$ maps variables to values.
We use $\scov{}_{[X,X]}$ to denote the variance of marginal variable $X$, and $\scov{}_{[X,.]}$, $\scov{}_{[.,X]}$ to denote the covariance vectors of $X$ with other variables in the state's multivariate Gaussian distribution.
Similarly, $\scov{}_{[\vec {X},\vec {X}]}$, and $\scov{}_{[\vec {X},\vec {Y}]}$ denote the covariance matrix of the sub-vector $\vec {X}$ of a multivariate Gaussian, and the covariance matrix between sub-vectors $\vec {X},\vec {Y}$ of a multivariate Gaussian, respectively.
We use $\mu_X$ to denote the mean of $X$, and $\smean_{\vect X}$ for the mean vector of $\vect X$.
\looseness -1

\begin{figure}[p]
\begin{mathpar}
\inferrule*[left={(\ruleexpvar)}]
{
  \sstore(x) = c
}
{
  \langle x, \s \rangle \rightarrow_{e} c
}
\hspace{11mm}
\inferrule*[left={(\ruleexpop)}]
{
  \langle e_{0}, \s \rangle \rightarrow_{e} c_0 \\
  \langle e_{1}, \s \rangle \rightarrow_{e} c_1
}
{
  \langle e_{0} \oplus e_{1}, \s \rangle \rightarrow_{e} c_0 \oplus c_1
}\hfill\strut
\\
\inferrule*[left={ (\ruleexpconst)}]
{ }
{
  \langle c, \s \rangle \rightarrow_{e} c
}
\hspace{11mm}
\inferrule*[left=(\ruledetassig)]
{
  e \rightarrow_e c
}
{
  \langle x = e, \s \rangle \rightarrow_{s} \langle \vec {\mu}, \matr{\Sigma}, \sstore[ x \mapsto c] \rangle
}\hfill\strut
\\
\inferrule*[
  left=(\rulepassigind),
  right=\(c_2 > 0\)
]
{
  \langle e_i, \s \rangle \to_e  c_i \text{ for } i = 1,2
  \enspace
  \smean' = {\begin{bmatrix} \smean \\ c_1 \end{bmatrix}}
  \enspace
  \scovprime = {\begin{bmatrix}
      \scovprime & \kern-4pt \vec 0 \\
      \vec 0     & \kern-4pt c_2
  \end{bmatrix}}
}
{
  \scalebox{1.2}{\strut}
  \langle X = \texttt{Normal}(e_1, e_2),\s \rangle \to_s
  \langle \smean', \scovprime, \sstore  \rangle
}
\hfill\strut

\inferrule*[
  left=(\rulepassigdep),
  right=\(c_3 > 0\)
]
{
  \scalebox{1.2}{\strut}\langle e_i, \s \rangle \rightarrow_{e} c_i \, \text{for} \, i = 1..3
  \\
  \smean' = {\begin{bmatrix}
      \smean \\ c_1 \mu_Y + c_2
  \end{bmatrix}}\scalebox{1.5}\strut
  \quad
  \scovprime = {\begin{bmatrix}
      \scov{}            &         c_1 \scov{}_{[.,Y]} \\
      c_1  \scov{}_{[Y,.]} &  c_1^2  \scov{}_{[Y,Y]} + c_3
  \end{bmatrix}}
}
{
  \scalebox{1.2}{\strut}
  \langle X = \texttt{Normal}(e_1 * Y + e_2, e_3),\s \rangle \rightarrow_{s} \langle \smean', \scovprime, \sstore  \rangle
}
\hfill\strut

\inferrule*[
  left=(\ruleopcond),
]
{
  \langle e, \s \rangle \rightarrow_{e} c   \hspace{5pt}
  \enspace
  \smean = {\begin{bmatrix}
             \smean_a \\
             \mu_X
         \end{bmatrix}}
  \enspace
  \scov{} = {\begin{bmatrix}
              \scov{}_a        & \scov{}_{[.,X]}  \\
              \scov{}_{[X,.]}  & \scov{}_{[X,X]}
          \end{bmatrix}}
  \\
  \smean' = \smean_a + (c - \mu_X) / \scov{}_{[X,X]} \scov{}_{[.,X]}
  \quad
  \scovprime = \scov{}_a - 1\scalebox{1.5}{\strut} / \scov{}_{[X,X]} \scov{}_{[.,X]} \scov{}_{[X,.]}
  %
}
{
  \langle  \texttt{condition}(X, e),\langle \smean,\scov{},\sstore \rangle \rangle \rightarrow_{s} \langle \smean', \scovprime, \sstore \rangle
}
\hfill\strut

\vspace{-3mm}

\inferrule*[left=(\ruleseq)\!]
{
  \langle \stmt_0, \s \rangle \! \to_s \! \s''
  \enspace
  \langle \stmt_1, \s'' \rangle \! \to_s \! \s'
}
{
  \langle \texttt{$\stmt_0$; $\stmt_1$}, \s \rangle \rightarrow_{s} \s'
}
\hfill
\inferrule*[left=(\rulereturn)\!]
{
  \langle \stmt, \s \rangle \! \to_s \!
  \langle
  {\begin{bmatrix}
    \vec \mu_{\vec {X}} \\
    \vec \mu_{\vec {Y}}
  \end{bmatrix}},
  {\begin{bmatrix}
    \scov{}_{[\vec {X},\vec {X}]} \scov{}_{[\vec {X},\vec {Y}]} \\
    \scov{}_{[\vec {Y},\vec {X}]} \scov{}_{[\vec {Y},\vec {Y}]}
  \end{bmatrix}},
  \sstore
  \rangle
}
{
  \langle \stmt \texttt{;\,return} \, \vec {X}, \s \rangle \rightarrow_p (\vec {\mu}_{\vec {X}},\scov{}_{[\vec {X},\vec {X}]})
}
\\
\inferrule*[left=(\ruleforb)\!]
{
  v_r \leq 0
}
{
  \langle \texttt{for\,\(x\)\,in\,range\,$v_r$ $\stmt$}, \s \rangle \! \to_s \! \s
}%
\hfill
\inferrule*[
  left=(\rulefori)\!
]
{
  v_r > 0 \quad
  \langle s',
  \s \rangle \! \to_s \! \s'
  \\\\
  s' \! = \! \stmt; x\,\texttt{=}\, x\texttt{+}1;\scalebox{1.2}{\strut} \texttt{for}\,x\,\texttt{in\,range}\,v_r \!\!-\!\! 1 \; \stmt
}
{
  \langle \texttt{for\;\(x\)\;in\,range\;$v_r$ $\stmt$}, \s \rangle \rightarrow_{s} \s'
}
\vspace{3.5mm}
\\
\inferrule*[left=(\ruleoprvs)]
{
  \smean' \! = \! {\begin{bmatrix}
              \smean \\
              \mu_Y + \mu_Z
          \end{bmatrix}}
  \quad
  \scovprime \! = \! {\begin{bmatrix}
      \scov{}  &  \scov{}_{[.,Y]} + \scov{}_{[.,Z]} \\
      \scov{}_{[Y,.]} \! + \! \scov{}_{[Z,.]} &
      \scov{}_{[Y,Y]} \! + \! \scov{}_{[Z,Z]} \! + \! \scov{}_{[Y,Z]} \! + \! \scov{}_{[Z,Y]}
             \end{bmatrix}}
}
{
  \langle X = Y + Z, \s \rangle \rightarrow_{s} \langle \smean', \scovprime, \sstore \rangle
}
\\
\ncondrule{\footnotesize \ruleopplusminus}
{ \footnotesize
  \oplus \in \{+,-\} \quad
  \langle e, \s \rangle \to_e c \strut \\
  \smean' \! = \! \begin{bmatrix}
              \smean \\
              \mu_Y \! \oplus \! c
            \end{bmatrix}

  \enspace
  \scovprime \! = \! \begin{bmatrix}
    \scov{}         & \kern-7pt \scov{}_{[.,Y]} \\
    \scov{}_{[Y,.]} & \kern-7pt \scov{}_{[Y,Y]}
  \end{bmatrix}
}
{ \footnotesize
  \langle X = Y \oplus e, \s \rangle \rightarrow_{s} \langle \smean', \scovprime, \sstore \rangle
}
\hfill
\ncondrule{\footnotesize \ruleopmuldiv}
{ \footnotesize
  \oplus \in \{*,/\} \quad
  \langle e, \s \rangle \to_e c \strut \\
  \smean' \! = \! \begin{bmatrix}
              \smean \\
              \mu_Y \! \oplus \! c
            \end{bmatrix}
  \enspace
  \scovprime \! = \! \begin{bmatrix}
    \scov{}                  & c \! \oplus \! \scov{}_{[.,Y]}  \\
    c \! \oplus \! \scov{}_{[Y,.]\kern-7pt} & c^2 \! \oplus \! \scov{}_{[Y,Y]}
  \end{bmatrix}
}
{ \footnotesize
  \langle X = Y \! \oplus \! e, \s \rangle \rightarrow_{s} \langle \smean', \scovprime, \sstore \rangle
}
\end{mathpar}
\caption{Operational Semantics rules; \s stands for a tuple $\langle \smean, \scov{}, \sstore \rangle$.}
\label{fig:semantics-rules}
\end{figure}

\begin{definition}[Semantics]
\label{def:semantics}
The semantics is given by the relations $\to_e : \expr \times \s \to \Reals$, $\to_s : \stmt \times \s \to \s$ and $\to_p : \pgrm \times \s \to \Reals^{n \times 1} \times \Reals^{n \times n}$ for expressions $\expr$, statements $\stmt$, and programs $\pgrm$, respectively, as defined in~\cref{fig:semantics-rules}.
\end{definition}

\noindent
The rules for expressions, (deterministic) assignments, sequence of statements, and for-loops are standard, and they do not change the state's multivariate Gaussian distribution.
We omit their details.
Programs finish with a \texttt{return} instruction.
It returns the mean-vector $\vec {\mu}_a$ and covariance matrix $\scov{}_{aa}$ of the specified sub-vector of the state's multivariate Gaussian.
In what follows, we focus on the rules manipulating the state's multivariate Gaussian distribution.

There are two types of probabilistic assignments: independent and linearly dependent.
In both cases the multivariate Gaussian distribution in the state is extended with a new variable, and consequently the mean vector ($\smean$) and covariance matrix ($\scov{}$) increase their dimension.
Independent assignments (\rulepassigind) add to the mean vector the mean of the distribution.
The covariance matrix is also updated with two $\vec {0}$ vectors indicating the new variable is not correlated with existing variables, and the variable's variance is added to the diagonal of the matrix.
Dependent assignments (\rulepassigdep) add to the mean vector a mean computed as a linear combination with the mean of the dependent random variable $Y$.
The covariance matrix is extended with two vectors computed from the covariance of the dependent variable with other variables, $\scov{}_{[Y,.]}$, $\scov{}_{[.,Y]}$.
This is because the new variable depends on $Y$ and consequently on all the variables that $Y$ depends on.
The variance of the new variable is added to the diagonal of the matrix as a linear combination with the variance of $Y$.

\begin{example}
\label{ex:probabilistic-assignment}
  Consider the program $X = \texttt{Normal}(15,2) \texttt{; } Y = \texttt{Normal}(20,1) \texttt{; } Z = \texttt{Normal}(2X,1)$.
  The first assignment results in state $\smean = [15]$ and $\scov{} = [2]$.
  The second assignment updates the state into,

  \vspace{-4pt plus 1pt}

  $$\footnotesize
  \smean' = \begin{bmatrix}
              15 \\
              20
            \end{bmatrix}
  ~~~
  \scovprime = \begin{bmatrix}
                 2 & 0 \\
                 0 & 1
               \end{bmatrix}
  $$

  \noindent
  Note the zeros in the covariance coefficients as these variables are independent. Finally, the last probabilistic statement results in

  \vspace{-4pt plus 1pt}

  $$\footnotesize
  \smean'' = \begin{bmatrix}
              15 \\
              20 \\
              2 \cdot 15
            \end{bmatrix}
          = \begin{bmatrix}
              15 \\
              20 \\
              30
            \end{bmatrix}
  ~~~
  \scovdoubleprime = \begin{bmatrix}
                       2         & 0         & 2 \cdot 2  \\
                       0         & 1         & 2 \cdot 0  \\
                       2 \cdot 2 & 2 \cdot 0 & 2^2 \cdot 2 + 1
                     \end{bmatrix}
                   = \begin{bmatrix}
                       2 & 0 & 4  \\
                       0 & 1 & 0  \\
                       4 & 0 & 9
                     \end{bmatrix}
  $$

  \noindent
  Here we observe that the covariance between $Z$ and $X$ is updated with a non-zero value due to the dependency between variables, but the coefficients of $Y$, $Z$ are 0 as these variable remain independent. \qed
\end{example}

%
%
\noindent
Two rules (\ruleopplusminus) and (\ruleopmuldiv) define binary operations between random variables and values.
These rules always produce a random variable that is added to the state's multivariate Gaussian.
This is why the statement is combined with an assignment.
\looseness -1

When a value is added/subtracted to a random variable (\ruleopplusminus), a new random variable is added with its mean updated accordingly.
The new variable inherits the variance an covariances of $Y$.
For multiplication and division (\ruleopmuldiv), the mean is updated as before, but also the variance of the random variable, and the covariances with the dependent random variables.

\begin{example}
\label{ex:bop-randomvariables-values}

Consider the program $X = \texttt{Normal}(1,1) \texttt{; } Y = X + 2 \texttt{; } Z = Y * 2$.
After the first statement we have the state $\smean = [1]$ and $\scov{} = [1]$.
The second statement updates the state such that,

$$\footnotesize
  \smean' = \begin{bmatrix}
              1 \\
              3
            \end{bmatrix}
  ~~~
  \scovprime = \begin{bmatrix}
                 1 & 1 \\
                 1 & 1
               \end{bmatrix}
$$

\noindent
The last statements updates the state into
$$\footnotesize
  \smean'' = \begin{bmatrix}
               1 \\
              3 \\
              6
            \end{bmatrix}
  ~~~
  \scovdoubleprime = \begin{bmatrix}
                       1         & 1         & 2 \cdot 1\\
                       1         & 1         & 2 \cdot 1 \\
                       2 \cdot 1 & 2 \cdot 1 & 2^{2} \cdot 1 \\
               \end{bmatrix}
                = \begin{bmatrix}
                      1 & 1 & 2 \\
                      1 & 1 & 2 \\
                      2 & 2 & 4
                    \end{bmatrix}
$$
\end{example}

%
\paragraph{Sum of random variables.}
The sum of two Gaussian random variables (\ruleoprvs) adds a new random variable to the multivariate Gaussian.
The mean of the new random variable is the sum of the means of the operands.
The covariance of the resulting random variable is the sum of the covariances of the operands with other variables, i.e., the new variable depends on all the variables that the operands depend on.
The variance is the sum of the variances of the operands, and the covariances of the operands.
%

\begin{example}
\label{ex:sum-random-variables}
  Consider the program $X = \texttt{Normal}(15,2) \texttt{; } Y = \texttt{Normal}(2,1) \texttt{; } Z = X + Y$.
  After the second assignment we have

  $$\footnotesize
  \smean' = \begin{bmatrix}
              15 \\
              2
            \end{bmatrix}
  ~~~
  \scovprime = \begin{bmatrix}
                 2 & 0 \\
                 0 & 1
               \end{bmatrix}
  $$

  \noindent
  Thus, the third assignment updates the state's multivariate Gaussian into
  $$\footnotesize
  \smean'' = \begin{bmatrix}
               15     \\
               2      \\
               15 + 2
             \end{bmatrix}
           = \begin{bmatrix}
               15 \\
               2  \\
               17
             \end{bmatrix}
 ~~~
 \scovdoubleprime = \begin{bmatrix}
                      2     & 0     & 2 + 0         \\
                      0     & 1     & 0 + 1         \\
                      2 + 0 & 0 + 1 & 2 + 1 + 0 + 0
                    \end{bmatrix}
                  = \begin{bmatrix}
                      2 & 0 & 2 \\
                      0 & 1 & 1 \\
                      2 & 1 & 3
                    \end{bmatrix}
  $$
\end{example}

\paragraph{Conditions.}
For conditioning (\ruleopcond), we use \cref{thm:conditions} introduced in~\cref{sec:gaussian-background}.
As a result of conditioning, the observed variable is removed from the mean vector and covariance matrix.
Note that, despite \ruleopcond\ applying to the newest random variable, we may perform an affine transformation using a permutation matrix that swaps the order of random variables.
Thus, conditions may refer to any variable in the multivariate Gaussian.

\begin{example}
  Let $\smean''$, and \scovdoubleprime\ be those in the final state of the program in~\cref{ex:sum-random-variables}.
  Suppose that we extend the program with the statement $\texttt{condition}(Z, 1)$.
  The resulting multivariate Gaussian is updated as

  $$\footnotesize
  \smean''' = \begin{bmatrix} 15 \\ 2 \end{bmatrix} + \frac{1-17}{3} \cdot \begin{bmatrix} 2 \\ 1 \end{bmatrix}
            = \begin{bmatrix} 15 \\ 2 \end{bmatrix} + \begin{bmatrix} -32/3 \\ -16/3 \end{bmatrix}
            = \begin{bmatrix} 13/3 \\ -10/3 \end{bmatrix}
  $$
  $$\footnotesize
  \scovtripleprime = \begin{bmatrix} 2 & 0 \\ 0 & 1 \end{bmatrix} - \frac{1}{3} \cdot \begin{bmatrix} 2 \\ 1 \end{bmatrix} \cdot \begin{bmatrix} 2 & 1 \end{bmatrix}
                   = \begin{bmatrix} 2 & 0 \\ 0 & 1 \end{bmatrix} - \frac{1}{3} \cdot \begin{bmatrix} 4 & 2 \\ 2 & 1 \end{bmatrix}
                   = \begin{bmatrix} 8/3 & -2/3 \\ -2/3 & 2/3 \end{bmatrix}
  $$
\end{example}

\noindent
Recall that covariances may be negative as the covariance matrix is positive definite (cf.~\cref{sec:gaussian-background}).

\subsection{Soundness and Termination}
\label{sec:properties}

In what follows, we show that the semantics rules in~\cref{fig:semantics-rules} are \emph{sound}, and that the inference engine always \emph{terminates} for well-formed programs.

We establish \emph{soundness} of our engine by ensuring that all program statements perform a closed-form transformation on the state's multivariate Gaussian distribution.
Lemmas (\ref{lemma:sum_rvs}-\ref{lemma:cond}) assert the soundness of each of the rules in $\to_s$ (cf.~\cref{fig:semantics-rules}).
For example, below we show the proof of the lemma for sum of random variables (i.e., \ruleoprvs) whose soundness is based on the affine transformation property of multivariate Gaussian distributions (cf.~\cref{thm:affine}).
The lemma asserts that the distribution resulting from executing the program statement is a well-formed multivariate Gaussian and also that the newly introduced variable is distributed as the sum of the operands.
We refer interested readers to~\cref{sec:proofs_lemmas} for the proofs of the remaining lemmas.
Again, we omit the soundness details of deterministic statements and expressions as they are standard.
The soundness of the \rulereturn\ rule follows from lemmas~(\ref{lemma:sum_rvs}-\ref{lemma:cond}) and~\cref{thm:marginals}.

\begin{lemma}[Sum random vars.]%
  \label{lemma:sum_rvs}
  Let $[X_1, X_2, \ldots]^\intercal \sim \gauss{\smean}{\scov{}}$.
  For all states $\s = \langle \smean, \scov{}, \sstore \rangle$, if $\langle Y =  X_i + X_j, \s \rangle \to_s \langle \smean', \scovprime, \sstore \rangle$ and $[X_1,X_2,\ldots,X_i + X_j]^\intercal \sim \gauss{\smean''}{\scovdoubleprime}$, then $[X_1,X_2,\ldots,Y]^\intercal \sim \gauss{\smean'}{\scovprime}$ and $\gauss{\smean'}{\scovprime} = \gauss{\smean''}{\scovdoubleprime}$.
\end{lemma}

\begin{proof}
  Let $\matr{A}$ be a $m \times n$ projection matrix where $\matr{A}_n = I_n$ where $I_n$ is a $n \times n$ identity matrix, $\matr{A}_{n+1}[i] = \matr{A}_{n+1}[j] = 1$ and $\matr{A}_{n+1}[k] = 0$ for $i,j \not= k$.
  Let $\vec {b} = \vec {0}$.
  Then, by matrix multiplication $[X_1, X_2, \ldots, X_i + X_j]^\intercal = \matr{A}\vec {X} + \vec {b}$.
  By~\cref{thm:affine}, $[X_1, X_2, \ldots, X_i + X_j]^\intercal \sim \gauss{\matr{A}\vect{\mu} + \vect{b}}{\matr{A}\scov{}\matr{A}^\intercal}$.
  By matrix multiplication
  \begin{equation}
    \begin{split}
    \matr{A}\vec {\mu} + \vec {b} &=
    \begin{bmatrix}
      \vec {\mu} \\
      \mu_{X_i} + \mu_{X_j}
    \end{bmatrix}
    \\
    \matr{A}\scov{}\matr{A}^\intercal &=
    \begin{bmatrix}
      \scov{} & \scov{}_{[.,X_i]} + \scov{}_{[.,X_j]} \\
      \scov{}_{[X_i,.]} + \scov{}_{[X_j,.]} & \scov{}_{[X_i,X_i]} + \scov{}_{[X_j,X_j]} + \scov{}_{[X_i,X_j]} + \scov{}_{[X_j,X_i]}
    \end{bmatrix}
    \end{split}
    \label{eq:affine_sum_rvs_main}
  \end{equation}
  By~\cref{def:semantics} and \cref{eq:affine_sum_rvs_main} we have that $\smean' = \matr{A}\vec {\mu} + \vec {b}$ and $\scovprime = \matr{A}\scov{}\matr{A}^\intercal$.
  Thus, $[X_1, X_2, \ldots, Y]^\intercal \sim \gauss{\smean'}{\scovprime}$ and $\gauss{\smean'}{\scovprime} = \gauss{\matr{A}\vect{\mu} + \vect{b}}{\matr{A}\scov{}\matr{A}^\intercal}$.
\end{proof}

\begin{lemma}[Independent assignment]
  Let $[X_1, X_2, \ldots]^\intercal \sim \gauss{\smean}{\scov{}}$. For all states $\s = \langle \smean, \scov{}, \sstore \rangle$, if $\langle Y = \texttt{\upshape Normal}(e_1, e_2), \s \rangle \to_s \langle \smean', \scovprime, \sstore \rangle$ and $\langle \s, e_i \rangle \to_e c_i$, then $[X_1,X_2, \ldots, Y]^\intercal \sim \gauss{\smean'}{\scovprime}$ and $Y \sim \gauss{c_1}{c_2}$ and $Y, X_i$ are independent.
  \label{lemma:ind_asg}
\end{lemma}

\begin{lemma}[Dependent assignments]%
  \label{lemma:cond_asg}
  Let $[X_1, X_2, \ldots]^\intercal \sim \gauss{\smean}{\scov{}}$. For all states $\s = \langle \smean, \scov{}, \sstore \rangle$, if $\langle Y = \texttt{\upshape Normal}(e_1*X_i + e_2, e_3), \s \rangle \to_s \langle \smean', \scovprime, \sstore \rangle$ and $\langle \s, e_i \rangle \to_e c_i$, then $[X_1, X_2, \ldots, Y]^\intercal \sim \gauss{\smean'}{\scovprime}$ and $Y \mid X_i \sim \gauss{c_1X_i + c_2}{c_3}$.
\end{lemma}

\begin{lemma}[Binary operations with values]
  Let $[X_1, X_2, \ldots]^\intercal \sim \gauss{\smean}{\scov{}}$ and $\oplus \in \{+,-,*,/\}$.
  For all states $\s = \langle \smean, \scov{}, \sstore \rangle$, if $\langle Y = X_i \, \oplus \, e, \s \rangle \rangle \to_s \langle \smean', \scovprime, \sstore \rangle$ and $\langle \s, e \rangle \to_e c$ and $[X_1, X_2, \ldots, X_i \oplus c]^\intercal \sim \gauss{\smean''}{\scovdoubleprime}$, then $[X_1,$ $ X_2, \ldots, Y]^\intercal \sim \gauss{\smean'}{\scovprime}$ and $\gauss{\smean'}{\scovprime} = \gauss{\smean''}{\scovdoubleprime}$.
\end{lemma}

\begin{lemma}[Conditioning]
  Let $[\vect{X}^\intercal_a, Y]^\intercal \sim \gauss{\smean}{\scov{}}$.
  For all states $\s = \langle \smean, \scov{}, $ $\sstore \rangle$, if $\langle \s, e \rangle \to_e c$ and $\langle \texttt{\upshape condition}(Y, c), \s \rangle \to_s \langle \smean', \scovprime, \sstore \rangle$, then $\vec {X}' \sim \gauss{\smean'}{\scovprime}$ and $\vec {X}' = \vec {X}_a \mid Y = c$.
  \label{lemma:cond}
\end{lemma}

\noindent
Readers familiar with semantics of probabilistic programs will note that these lemmas define the cases of a pushfoward measure semantics \textit{\`a la} Kozen \cite{barthe_katoen_silva_2020} on the program statements in our language. A standard induction on the structure of well-formed programs establishes soundness of our inference engine.

We also establish termination of our inference engine.

\begin{lemma}[Termination]
  Given a well-formed program, the process of computing the resulting multivariate Gaussian always terminates.
\end{lemma}

\begin{proof}
  Well-formed programs are unbounded but finite sequences of program statements.
  Thus, to prove termination, it suffices to prove that each program statement is evaluated in finite time.
  Expressions (\ruleexpvar, \ruleexpconst, \ruleexpop) and deterministic statements (\ruledetassig, \ruleseq) can be resolved in constant time.
  For-loops (\ruleforb,\rulefori) are bounded and can be unfolded in linear time in the number of iterations.
  Probabilistic assignments (\rulepassigind, \rulepassigdep) extend the mean vector with one element and the covariance matrix with a row and column vectors.
  Both operations can be performed in linear time in the number of variables of the program.
  Binary operations between random variables and values (\ruleopplusminus, \ruleopmuldiv), summation of random variables (\ruleoprvs) and conditioning (\ruleopcond) are computed as a sequence of matrix multiplication operations.
  All these operations can be computed in polynomial time in the size of the program state, which is never larger than quadratic in the number of variables in the program.
  Return (\rulereturn) performs a lookup in the mean vector and covariance matrix, which can be computed in constant time.
\end{proof}

\noindent
Our inference engine not only terminates for all well-formed programs, but, most importantly, it is \emph{very} efficient.
In~\cref{sec:eva:scalability}, we study the scalability of the inference engine, and we show that it can efficiently analyze systems with thousands of random variables.
Moreover, we show that our method scales much better than existing tools for the family of programs captured by our language.


\section{Case Study: Privacy Risks in Public Statistics}
\label{sec:case-study}

We analyze a program computing statistics on a database containing incomes for different genders and age groups.
The purpose of this case study is to demonstrate the applicability of our approach in a real-life example.
Average incomes are available through public national statistics banks~\cite{statsdk,statsnz,us_census}, which makes information available to attackers.
Leakage of private information and database reconstruction attacks are known issues (e.g., in US census data~\cite{DBLP:journals/cacm/GarfinkelAM19}).
We use our inference engine to quantify the increase of attacker knowledge, as she gradually obtains statistics from a database.
We also analyze a differentially private~\cite{dp} mechanism in this setting.
The case study uses a small database, but in~\cref{sec:eva:scalability} we show that our inference engine scales to databases with thousands of individuals.

\paragraph{Releasing public statistics.}
Consider a data analyst that releases average statistics on population income for different age groups and genders.
An attacker with access to the statistics attempts to learn the income of an individual in the database.
We consider the synthetic data shown in \cref{tab:income_data_prior} (left).
The table shows the income for 40 individuals in different age groups and genders.
We consider 3 different cases.
\begin{enumerate*}[label=\textit{Case} \arabic*)]
\item 
  the attacker obtains the average income for males in the age group 21-30.
\item  
  the attacker also obtains the average income for all people in the age group 21-30.
\item  
  the attacker also obtains the average income for all males.
\end{enumerate*}
In all cases the attacker attempts to learn the income of the first male in database in the age group 21-30. We note that the observations made by the attacker are independent, so it does not matter in which order the attacker obtains the observations.

\begin{table}[t!]
\scalebox{0.70}{
\begin{tabular}{c c c c c  }
\hline
Age group & \multicolumn{2}{c}{Males} & \multicolumn{2}{c}{Females}   \\
\hline
 \multirow{5}{*}{21-30}  &
 500~ &  480~ \vline & ~470~ &   410 \\
 &   460~ &  430~  \vline& ~420~ &   450 \\
&   490~ &  510~ \vline& ~460~ &   410 \\
 &   440~ &  480~ \vline& ~510~ &   310 \\
 &   520~ &  410~  \vline& ~370~ &   440 \\
\hline

 \multirow{5}{*}{31-40} &
 550~  &  410~  \vline&  ~450~ &   500   \\
 &  490~  &  580~  \vline&  ~520~ &   530   \\
 &  530~  &  420~  \vline&  ~510~ &  600   \\
 &  590~  &  400~  \vline&  ~620~ &   390   \\
 &  680~  &  510~  \vline&  ~550~ &   390   \\
\hline

\multirow{5}{*}{41-50}&
600~ &  540~  \vline&   ~590~  &  640 \\
&   640~ &  590~  \vline&   ~540~ &  580 \\
&   580~ &  620~  \vline&   ~740~  &  540 \\
&   340~ &  510~  \vline&   ~140~  &  830 \\
&   620~ &  660~  \vline&   ~540~  &  740 \\
\hline

\multirow{5}{*}{51-60} &
700~  & 680~  \vline&  ~690~ &  680 \\
 &  740~  & 640~  \vline&  ~720~ &  780 \\
 &  590~  & 650~  \vline&  ~680~ &  580 \\
 &  770~  & 630~  \vline&  ~590~ &  730 \\
 &  540~  & 840~  \vline&  ~640~ &  980 \\
\hline
\end{tabular}
}
\hspace{2ex}
\scalebox{0.70}{
\begin{tabular}{c c c c c  }
\hline
Age group & \multicolumn{2}{c}{Males} & \multicolumn{2}{c}{Females}   \\
\hline
 \multirow{5}{*}{21-30}
 & $\gauss{480}{100}$ ~ &  $\gauss{490}{100}$~ \vline & ~$\gauss{450}{100}$ ~&   $\gauss{430}{100}$ \\

 &   $\gauss{440}{100}$ ~ &  $\gauss{420}{100}$~  \vline& ~$\gauss{410}{100}$ ~&   $\gauss{430}{100}$ \\

&   $\gauss{490}{100}$ ~&  $\gauss{490}{100}$~  \vline& ~$\gauss{470}{100}$ ~&   $\gauss{400}{100}$ \\

 &   $\gauss{520}{100}$ ~&  $\gauss{490}{100}$~  \vline& ~$\gauss{490}{100}$ ~&   $\gauss{330}{100}$ \\

 &   $\gauss{470}{100}$ ~&  $\gauss{400}{100}$~  \vline& ~$\gauss{350}{100}$ ~&   $\gauss{400}{100}$ \\
\hline

 \multirow{5}{*}{31-40}
 &   $\gauss{500}{100}$ ~&  $\gauss{410}{100}$~  \vline& ~$\gauss{400}{100}$ ~&   $\gauss{450}{100}$ \\
 &   $\gauss{470}{100}$ ~&  $\gauss{490}{100}$~  \vline& ~$\gauss{490}{100}$ ~&   $\gauss{480}{100}$ \\
 &   $\gauss{500}{100}$ ~&  $\gauss{410}{100}$~  \vline& ~$\gauss{500}{100}$ ~&   $\gauss{550}{100}$ \\
 &   $\gauss{540}{100}$ ~&  $\gauss{410}{100}$~  \vline& ~$\gauss{590}{100}$ ~&   $\gauss{350}{100}$ \\
 &   $\gauss{500}{100}$ ~&  $\gauss{400}{100}$~  \vline& ~$\gauss{510}{100}$ ~&   $\gauss{360}{100}$ \\
\hline

\multirow{5}{*}{41-50}&
$\gauss{580}{100}$ ~ &  $\gauss{530}{100}$~ \vline & ~$\gauss{550}{100}$ ~&   $\gauss{490}{100}$ \\
 &   $\gauss{590}{100}$ ~ &  $\gauss{510}{100}$~  \vline& ~$\gauss{520}{100}$ ~&   $\gauss{500}{100}$ \\
&   $\gauss{560}{100}$ ~&  $\gauss{590}{100}$~  \vline& ~$\gauss{650}{100}$ ~&   $\gauss{480}{100}$ \\
 &   $\gauss{280}{100}$ ~&  $\gauss{500}{100}$~  \vline& ~$\gauss{150}{100}$ ~&   $\gauss{790}{100}$ \\
 &   $\gauss{580}{100}$ ~&  $\gauss{600}{100}$~  \vline& ~$\gauss{510}{100}$ ~&   $\gauss{700}{100}$ \\
\hline

\multirow{5}{*}{51-60}
& $\gauss{680}{100}$ ~ &  $\gauss{570}{100}$~    \vline & ~$\gauss{680}{100}$ ~&   $\gauss{670}{100}$ \\
 &   $\gauss{620}{100}$ ~ &  $\gauss{610}{100}$~  \vline& ~$\gauss{690}{100}$ ~&   $\gauss{700}{100}$ \\
&   $\gauss{600}{100}$ ~&  $\gauss{570}{100}$~  \vline& ~$\gauss{630}{100}$ ~&   $\gauss{570}{100}$ \\
 &   $\gauss{700}{100}$ ~&  $\gauss{600}{100}$~  \vline& ~$\gauss{500}{100}$ ~&   $\gauss{670}{100}$ \\
 &   $\gauss{520}{100}$ ~&  $\gauss{770}{100}$~  \vline& ~$\gauss{600}{100}$ ~&   $\gauss{770}{100}$ \\
\hline
\end{tabular}
}
\caption{
  \textit{Left}: Incomes pr. year in DKK for different age groups and genders. The numbers have been scaled down by a factor of 1000.
  \textit{Right}: Priors used in the experiments. The means are scaled down by a factor of 1000.
}
\label{tab:income_data_prior}
\end{table}

To understand the privacy risks for these 3 cases we use the Privug method described in \Cref{sec:preliminaries}.
In the following code snippet we show how to model the steps for case 1 using our inference engine.
We model the program using our syntax (step 2).
\begin{lstlisting}[ label={lst:1} ]
def agg():
    male_21_30 = [Normal(480_000, 100), ....]
    male_21_30_total = male_21_30[0] + male_21_30[1] ...
    male_21_30_average = male_21_total/10
    condition("male_21_30_average", 472_000)
    return male_21_30[0]
\end{lstlisting}
The array in line 2 contains the priors of the income for the individuals in the database---\cref{tab:income_data_prior} (right) shows the complete list---denoted as $P(I_i)$.
They represent the possible incomes that the attacker considers possible before making any observations (step 1).
The victim is the first male in the 21-30 age group, $P(I_1)$.
In lines 2-3, we compute the average income of this each group, which defines $P(O|I_i)$.
In line 5, we add the attacker observation in the condition statement (step 3).
Finally, in line 6 we return the posterior distribution of the victim $P(I_1|O)$, which represents the updated attacker knowledge (step 4).
Note that, for brevity, we omit the repetitive parts of the code. Also, arrays are syntactic sugar.
We refer interested readers to~\cite{gauss-privug-doi} for the complete source code.

\Cref{fig:risk} shows how attacker knowledge is updated in the 3 cases above, and how close it is to real victim data (vertical line). 
We plot the prior attacker knowledge $P(I_1)$, and for each case we plot the posterior distribution after conditioning on the output $P(I_1 | O)$.
The left plot shows the updated attacker knowledge without differential privacy. As the plot shows, the attacker knowledge gets closer to the actual income when obtaining more information. The most accurate attacker knowledge is case 3 where the attacker obtains several average statistics.
\paragraph{Releasing public statistics with differential privacy.}
Given the above results, the data analyst decides to use a differentially private mechanism~\cite{dp} to protect the individuals' privacy.
Differential Privacy (DP) is used in realistic settings for the release of public statistics.
Notably, it was used in the 2020 US Census as a result of privacy issues in previous US Census editions~\cite{DBLP:journals/cacm/GarfinkelAM19}.
Intuitively, if the privacy protection mechanism satisfies differential privacy, then the impact of an individual on the output of the program is negligible.
More precisely, differential privacy states that: a randomized mechanism $\mathcal{M} : \Inputs \to \Outputs$ is $(\epsilon,\delta)$-differentially private if for all $\mathcal{O} \subseteq \Outputs$, and neighboring inputs $i_1, i_2 \in \Inputs$, the following holds
$$
P(\mathcal{M}(i_1) \in \mathcal{O}) \leq \exp(\epsilon)P(\mathcal{M}(i_2) \in \mathcal{O}) + \delta.
$$
The neighboring relation between inputs depends on the input domain ($\Inputs$).
For instance, when it applies to datasets of $n$ natural numbers, $\mathbb{N}^n$, it is usually defined as the first norm $||i_x - i_y||_1$.
The parameter $\epsilon$ is often referred to as the \emph{privacy parameter}, and it is used to specify the required level of privacy.
The parameter $\delta$ is the probability of failure.
This parameter relaxes the definition of differential privacy.
It is used to specify the probability that pure differential privacy (i.e., with $\delta = 0$) does not hold.
This parameter may be used to, e.g., enable high utility gains while keeping a good level of privacy.
Both $\epsilon$ and $\delta$ are often determined empirically~\cite{Dwork.expose_epsilon.2019}.

We analyze a differentially private mechanism for the 3 cases presented above.
To this end, we apply the Gaussian mechanism~\cite{dp}, which adds Gaussian noise to the observable output ($\mathit{o}$) as $\mathit{o} + \gauss{0}{\std}$.
The parameter $\std$ is calculated as follows:
$
\std = 2\Delta^{2}\log(1.25/\delta) / \epsilon^{2}.
$
The sensitivity $\Delta \in \Reals$ denotes how much $\mathit{o}$ changes if computed in two datasets differing in at most 1 entry.
In our setting, it is  $ \Delta = (\mathit{max}_{\mathit{income}} - \mathit{min}_{\mathit{income}}) / \mathit{size}_{\mathit{DB}}$.
We set $ \delta = 1 / \mathit{size}_{\mathit{DB}}^{2}$---as usual for this query~\cite{dp}.
We set $\epsilon$ to 0.9---this is an arbitrary value, but it is common to use values < 1 in practice~\cite{Dwork.expose_epsilon.2019}.
Adding Gaussian noise is proven to satisfy $(\epsilon,\delta)$-differential privacy~\cite{dp}.
We remark that our method can be used to determine the values of $\epsilon$ and $\delta$ that satisfy high level privacy requirements.
For instance, privacy requirements specified as probability queries for a given individual or using quantitative information flow metrics~\cite{privug}.
The program implementing the Gaussian mechanism is shown in the following listing
\begin{lstlisting}[ label={lst:2} ]
def agg_dp():
    ...
    noise = Normal(0, 1442533240)
    male_21_30_average_dp = male_21_30_average + noise
    condition("male_21_30_average_dp", 472_000)
    return male_21_30[0]
\end{lstlisting}
We only show lines that change namely: line 3 where the noise distribution is defined, and line 6 where we add the noise to the output. The variance $\std$ of the noise distribution is calculated using the equation above.

The right plot in \cref{fig:risk} shows the updated attacker knowledge in the 3 cases.
We observe a decrease in privacy risks when using differential privacy; as the change in attacker knowledge is insignificant for all cases.
The plot shows that the impact of the victim's data on the released statistics is minuscule compared to the non-differentially private version of the output.

\begin{figure}[t!]
    \centering
    \includegraphics[width=\textwidth]{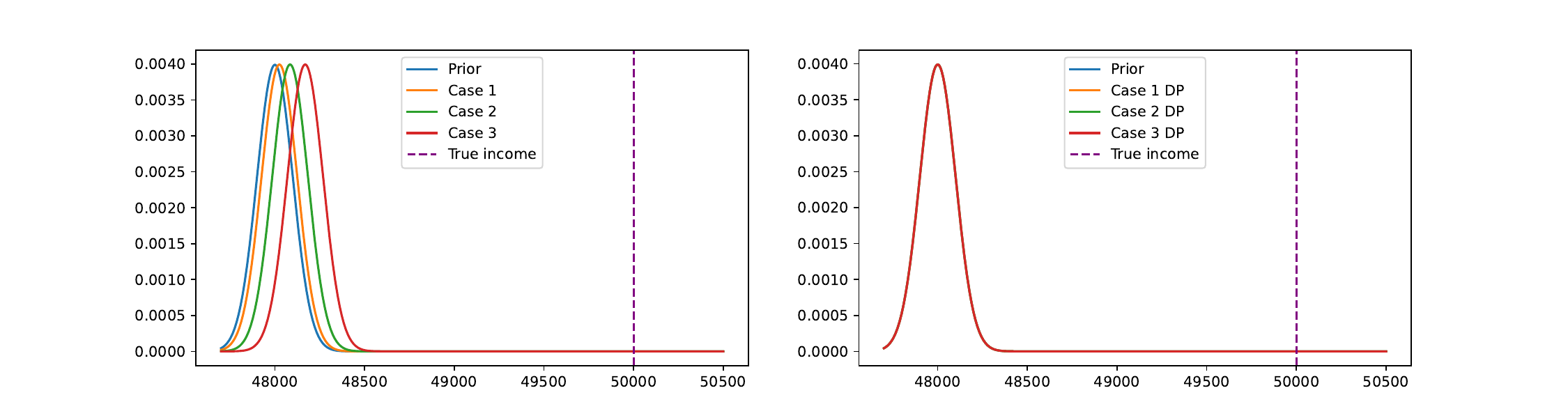}
    \caption{Updated attacker knowledge after adding observations. Incomes are scaled down by a factor of 10. \textit{Left}: Public stats. \textit{Right}: Public stats with DP.}
    \label{fig:risk}
\end{figure}

\paragraph{Information leakage metrics.}
In addition to inspecting the distributions of attacker knowledge in~\cref{fig:risk}, we show for demonstration how to compute two metrics for information leakage: KL-divergence and mutual information~\cite{elementsofinformationtheory.2006}.
Let $[P,Q]^\intercal \sim \gauss{\smean}{\scov{}}$.
KL-divergence is $\mathit{KL}(P,Q) = \log_{2} (\matr{\Sigma}^{1/2}_{[Q,Q]} / \matr{\Sigma}^{1/2}_{[P,P]}) + (\matr{\Sigma}_{[P,P]} + (\mu_{P} - \mu_{Q})^{2})  / 2\matr{\Sigma}_{[Q,Q]} - 1/2$, and mutual information is $I(P,Q) = $ $1/2 \log_2 (\matr{\Sigma}_{[P,P]} \matr{\Sigma}_{[Q,Q]} / |\matr{\Sigma}|)$.

\Cref{tab:qif_metrics_evaluation} shows the results.
The left shows the KL-divergence between prior and posterior attacker knowledge on the secret.
Intuitively, this is commonly understood as \emph{information gain}~\cite{burnham.anderson.klinformationgain.2002}.
We observe an increase in information gain from case 1 to 3 (both with and without differential privacy).
However, with differential privacy, information gain is virtually 0 for all cases.
\cref{tab:qif_metrics_evaluation} (right) shows mutual information between attacker knowledge on the secret and the output.
When mutual information between two random variables is 0, it means that the variables are independent.
Thus, a value of mutual information close to 0 indicates that the amount of information shared between secret and output is low.
We observe that mutual information decreases from case 1 to 3 (both with and without differential privacy).
This is due to the output containing information for a larger set of individuals (which minimizes the effect of the secret on the output).
As expected, mutual information is lower with differential privacy.
Admittedly, these metrics are hard to interpret in practice, but we remark that more important than the concrete values is their relative distance---it provides a quantitative mean to compare information leakage in different settings.

\begin{table}[t!]
  \begin{center}
    \begin{tabular}{ l@{\hskip 5ex} r@{\hskip 2ex} r@{\hskip 2ex} r@{\hskip 5ex} r@{\hskip 2ex} r@{\hskip 2ex} r }
      \toprule
      & \multicolumn{3}{c}{KL divergence} & \multicolumn{3}{c}{Mutual information} \\
      & Case 1 & Case 2 & Case 3 & Case 1 & Case 2 & Case 3  \\
      \midrule[0.15mm]
      Public stats &  312.762 &  3621.29  & 14374.13  &  4.17e-04  &  6.25e-05  &  7.81e-06  \\
      Diff. Priv.  &  1.55e-12 &    1.72e-12  &   1.81e-12  &  5.00e-12  &  3.14e-13  &  2.16e-14 \\
      \bottomrule
    \end{tabular}
  \end{center}
  \caption{\textit{Left}: KL-divergence between prior and posterior attacker knowledge on secret. \textit{Right}: Mutual information between secret and output random variables.}
  \label{tab:qif_metrics_evaluation}
  \label{tab:kl-divergence}
  \label{tab:mutual_information}
\end{table}


\section{Scalability Evaluation}
\label{sec:eva:scalability}
We evaluate the scalability of our exact inference engine proof-of-concept implementation.
The scalability of Bayesian inference engines mainly depends on the number of random variables.
Thus, we consider two synthetic benchmark programs with increasing number of variables.
The first computes the sum over an increasing number of variables $O = \sum^n_{i=1} X_i $.
We choose this benchmark as it was originally used to measure the scalability of Privug~\cite{privug}.
We compare the scalability of our engine to Privug MCMC using the NUTS~\cite{NUTS} sampler and the exact inference engine PSI~\cite{gehr_psi_2020}---PSI is the leading inference engine supporting the features of our language (cf.~\cref{sec:related}).
We instruct NUTS to draw 10000 samples in 2 chains---this number of samples produces an accurate posterior in this benchmark, see~\cite{privug}.
The second program performs the same computation but adds a condition statement $\texttt{\upshape condition}(O,c)$.
The purpose is to evaluate the scalability of our engine in more realistic settings for the case study in~\cref{sec:case-study}.
The evaluation run on a 4x2.80GHz cores machine with 16 GB RAM.

\Cref{fig:scalability} (left) shows the measured times for the first program. 
Execution time does not increase significantly when going from 100 to 700 variables using our engine.
On the other hand, PSI takes approximately 40 minutes when summing 700 variables.
Most notably, our engine greatly outperforms PSI---in this experiment, it was more than 40.000 times faster than PSI for 700 variables.
Privug MCMC exhibits better results, but our engine scales better.
As the number of variables increases, we observe a bigger gap between our engine and Privug MCMC---with our engine 6 times faster for 700.
It is noteworthy that our exact engine outperforms an approximate inference method.

\Cref{fig:scalability} (middle,right) focus on the scalability for larger systems.
The middle plot, shows that our engine can handle the first program with 70000 variables more efficiently that PSI for 700.
\Cref{fig:scalability} (right) mimics the case study (\cref{sec:case-study}).
We observe that the condition statement notably degrades the performance of our engine.
However, the running time for 5000 individuals is less than 40min.
We omitted PSI in this benchmark as conditioning would only decrease its performance, and the previous experiment showed its lower scalability w.r.t. our engine.
Privug MCMC is also omitted as a fair comparison requires determining the number of samples to draw to obtain an accurate posterior.

\begin{figure}[!t]
    \centering
    \includegraphics[width=1.0\textwidth]{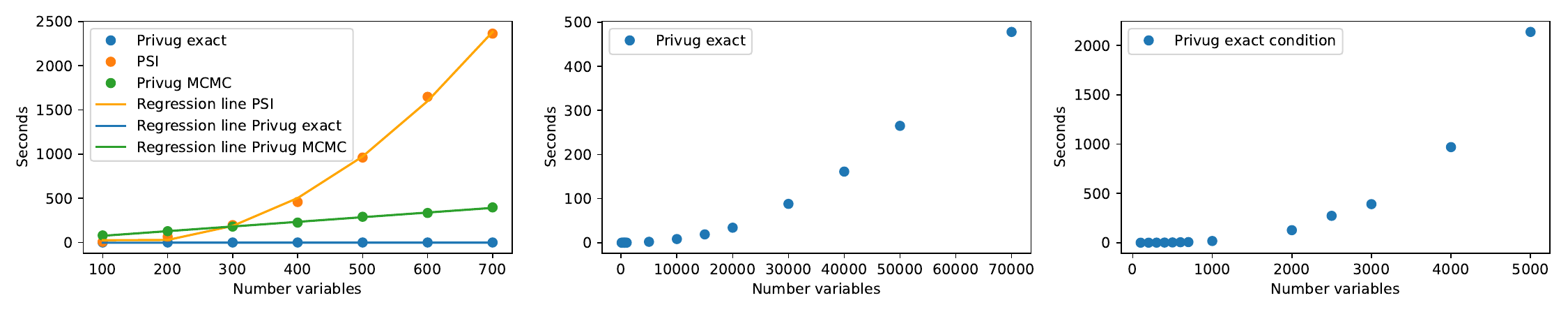}
    \caption{
      Execution time for our engine (privug-exact), Privug NUTS (privug-mcmc) and PSI.
      \textit{Left}: privug-mcmc, privug-exact, and PSI on $O = \sum^{n}_i X_i$. 
      \textit{Middle}: privug-exact on $O = \sum^{n}_i X_i$. 
      \textit{Right}: privug-exact on $\texttt{\upshape condition}(O,c)$.
    }
    \label{fig:scalability}
\end{figure}


\section{Related Work}
\label{sec:related}

%
The majority of existing methods to estimate privacy risks use sampling based techniques~\cite{privug,romanelli.leaves.2020,cherubin.fbleau.2019,chothia.leakwatch.2014,chothia.leakest.2013,DBLP:conf/csfw/ChothiaG11}.
In~\cite{privug}, Privug made use of MCMC algorithms to perform Bayesian inference, e.g., \textit{Metropolis-Hastings} or \emph{Hamiltonian Monte Carlo}~\cite{avi_pfeffer_practical_2016,greenberg_understanding_nodate,NUTS}.
Other sampling based methods target specific quantitative information flow metrics~\cite{qifbook.2020}---these metrics are supported by Privug~\cite{privug}, and hence by our engine.
LeakWatch/Leakiest~\cite{chothia.leakwatch.2014,chothia.leakest.2013} use program samples to estimate mutual information between secret inputs and public outputs.
Cherubin et al. and Romanelli et al.~\cite{cherubin.fbleau.2019,romanelli.leaves.2020}, use machine learning to compute metrics from the g-leakage family~\cite{qifbook.2020}.
These methods treat programs as black-boxes, so they can analyze any program, as opposed to our method that targets a subset of Python programs.
However, their accuracy guarantees are proven in the limit, i.e., assuming an infinite size sample.
In practice, samples are finite and it is often difficult to ensure that results are accurate; specially for programs with large number of variables (such as the ones in~\cref{sec:eva:scalability}).
On the contrary, our inference engine produces exact results.
This is crucial as a under-approximations could miss important privacy breaches.
Furthermore, the scalability evaluation shows that the inference engine scales better than MCMC-based Privug, which is one of the most scalable methods for this type of systems~\cite{privug}.
\looseness -1

%
There exist several works that use exact inference in the context of privacy risk analysis.
SPIRE~\cite{spire} uses the exact inference engine PSI~\cite{PSI,gehr_psi_2020} to model attacker knowledge and synthesize privacy enforcers.
PSI computes a symbolic representation of the joint probability distribution of a given program.
It can handle continuous and discrete random variables.
It targets a more expressive programming language than the subset of Python that our engine supports.
However, PSI scales poorly compared to our engine for programs that our engine supports (cf.~\cref{sec:eva:scalability}).
Hakaru~\cite{hakaru} and SPPL~\cite{SPPL} are exact inference engines---not used for privacy risk analysis.
We did not consider them in our evaluation because they do not handle some features of our language.
Hakaru cannot handle conditioning probability-zero events (as in~\cref{lemma:cond}).
SPPL does not support linear combination and sum of Gaussians (as in lemmas~\ref{lemma:cond_asg}, \ref{lemma:sum_rvs}).
QUAIL~\cite{QUAIL} computes mutual information between input and output variables.
It performs forward state exploration of a program to construct a Markov chain, which is then used to compute mutual information.
QUAIL works on discrete random variables.
Instead, our inference engine works on Gaussian (continuous) random variables and computes the posterior distribution that can be used to compute mutual information (cf.~\cref{sec:case-study}) and other quantitative information flow metrics~\cite{privug,qifbook.2020}.

Stein and Staton proposed a Gaussian-based semantics to study exact conditioning through the lens of category theory~\cite{stein.gauss.exact.2021}.
They do not study the use of the semantics for privacy risks quantification on a subset of Python programs, or evaluate the efficiency of the semantics.


\section{Conclusion}
\label{sec:conclusion}

We have presented an exact Bayesian inference engine for quantifying privacy risks in a subset of Python.
We have proven that our inference engine is sound.
We have presented an application of our engine to analyze privacy risks on public statistics; a realistic case study for national statistics agencies where privacy risks analysis is crucial.
We have also analyzed the impact of differential privacy on data release.
In the scalability evaluation, we have shown that our engine can analyze systems with thousands of random variables, and that it greatly outperforms existing tools.
All in all, this work provides a new point in the study of expressiveness vs performance.
Future work includes adapting our engine with underlying probabilistic models that capture more Python program statements, for instance Gaussian mixtures or the exponential family of probability distributions.
\looseness -1


\bibliographystyle{splncs04}
\bibliography{references}

\newpage
\appendix

\section{Proofs}
\label{sec:proofs_lemmas}

\setcounter{lemma}{0} 
\begin{lemma}[Sum random var.]
  Let $[X_1, X_2, \ldots]^\intercal \sim \gauss{\smean}{\scov{}}$. 
  For all states $\s = \langle \smean, \scov{}, \sstore \rangle$, if $\langle Y = X_i + X_j, \s \rangle \rangle \to_s \langle \smean', \scovprime, \sstore \rangle$ and $[X_1, X_2, \ldots, X_i + X_j]^\intercal \sim \gauss{\smean''}{\scovdoubleprime}$, then $[X_1, X_2, \ldots, Y]^\intercal \sim \gauss{\smean'}{\scovprime}$ and $\gauss{\smean'}{\scovprime} = \gauss{\smean''}{\scovdoubleprime}$.
\end{lemma}

\begin{proof}  
  Let $\matr{A}$ be a $m \times n$ projection matrix where $\matr{A}_n = I_n$ where $I_n$ is a $n \times n$ identity matrix, $\matr{A}_{n+1}[i] = \matr{A}_{n+1}[j] = 1$ and $\matr{A}_{n+1}[k] = 0$ for $i,j \not= k$.
  Let $\vect{b} = \vect{0}$.
  Then, by matrix multiplication $[X_1, X_2, \ldots, X_i + X_j]^\intercal = \matr{A}\vect{X} + \vect{b}$.
  By~\cref{thm:affine}, $[X_1, X_2, \ldots, X_i + X_j]^\intercal \sim \gauss{\matr{A}\vect{\mu} + \vect{b}}{\matr{A}\scov{}\matr{A}^\intercal}$.
  By matrix multiplication
  \begin{equation}
    \begin{split}
    \matr{A}\vect{\mu} + \vect{b} &= 
    \begin{bmatrix}
      \vect{\mu} \\
      \mu_{X_i} + \mu_{X_j}
    \end{bmatrix}
    \\
    \matr{A}\scov{}\matr{A}^\intercal &= 
    \begin{bmatrix}
      \scov{} & \scov{}_{[.,X_i]} + \scov{}_{[.,X_j]} \\
      \scov{}_{[X_i,.]} + \scov{}_{[X_j,.]} & \scov{}_{[X_i,X_i]} + \scov{}_{[X_j,X_j]} + \scov{}_{[X_i,X_j]} + \scov{}_{[X_j,X_i]}
    \end{bmatrix}
    \end{split}
    \label{eq:affine_sum_rvs}
  \end{equation}
  By~\cref{def:semantics} and \cref{eq:affine_sum_rvs} we have that $\smean' = \matr{A}\vect{\mu} + \vect{b}$ and $\scovprime = \matr{A}\scov{}\matr{A}^\intercal$.
  Thus, $[X_1, X_2, \ldots, Y]^\intercal \sim \gauss{\smean'}{\scovprime}$ and $\gauss{\smean'}{\scovprime} = \gauss{\matr{A}\vect{\mu} + \vect{b}}{\matr{A}\scov{}\matr{A}^\intercal}$.
\end{proof}

\begin{lemma}[Independent assignment]
  Let $[X_1, X_2, \ldots]^\intercal \sim \gauss{\smean}{\scov{}}$.
  For all states $\s = \langle \smean, \scov{}, \sstore \rangle$, if $\langle Y = \texttt{\upshape Normal}(e_1, e_2), \s \rangle \rangle \to_s \langle \smean', \scovprime, \sstore \rangle$ and $\langle \s, e_i \rangle \to_e c_i$, then $[X_1,X_2,\ldots,Y]^\intercal \sim \gauss{\smean'}{\scovprime}$ and $Y \sim \gauss{c_1}{c_2}$ and $Y, X_i$ are independent for all $X_i$.
\end{lemma}

\begin{proof}
  Let $[X_1, X_2, \ldots, Y]^\intercal \sim \gauss{\smean''}{\scovdoubleprime}$ where $Y \sim \gauss{c_1}{c_2}$ and $Y,X_i$ are independent for all $X_i$.
  Then, by~\cref{thm:independence}, $\scovdoubleprime_{[.,Y]} = \vect{0}$ and $\scovdoubleprime_{[Y,.]} = \vect{0}^\intercal$.
  Also, by~\cref{thm:marginals}, $\mu''_Y = c_1$ and $\scovdoubleprime_{[Y,Y]} = c_2$, and $\smean''_{[X_1, X_2, \ldots]} = \smean$ and $\scovdoubleprime_{[X_1, X_2, \ldots]} = \scov{}$.
  By~\cref{def:semantics}, $\smean' = \smean''$ and $\scovprime = \scovdoubleprime$.
  Thus, $[X_1,X_2,\ldots,Y]^\intercal \sim \gauss{\smean'}{\scovprime}$ and $Y \sim \gauss{c_1}{c_2}$ and $Y, X_i$ are independent for all $X_i$.
\end{proof}

\begin{lemma}[Dependent assignments]
  Let $[X_1, X_2, \ldots]^\intercal \sim \gauss{\smean}{\scov{}}$. For all states $\s = \langle \smean, \scov{}, \sstore \rangle$, if $\langle Y = \texttt{\upshape Normal}(e_1*X_i + e_2, e_3), \s \rangle \rangle \to_s \langle \smean', \scovprime, \sstore \rangle$ and $\langle \s, e_i \rangle \to_e c_i$, then $[X_1, X_2, \ldots, Y]^\intercal \sim \gauss{\smean'}{\scovprime}$ and $Y \mid X_i \sim \gauss{c_1X_i + c_2}{c_3}$.
\end{lemma}

\begin{proof}
  We instantiate~\cref{thm:linear} for a linear Gaussian combination of one random variable in $\gauss{\smean}{\scov{}}$ and show that the posterior equals $\gauss{\smean'}{\scovprime}$.
  Let $\vect{a}$ be column vector such that $\vect{a}[i] = c_1$ and $\vect{a}[j]=0$ for $i \not= j$, and $b=c_2$.
  Then, $\vect{a}^\intercal[X_1, X_2, \ldots]^\intercal + c_2 = c_1X_i + c_2$.
  By~\cref{thm:linear}, if $[X_1, X_2, \ldots]^\intercal \sim \gauss{\smean}{\scov{}}$ and $Y \mid X_i \sim \gauss{c_1X_i + c_2}{c_3}$, then $[X_1, X_2, \ldots, Y]^\intercal \sim \gauss{\smean''}{\scovdoubleprime}$ with
  \begin{equation}
    \mu''_Y = c_1\mu_{X_i} + c_2 
    ~~~ 
    \scovdoubleprime_{[Y,Y]} = c_1^2\scov{}_{[X_i,X_i]} + c_3
    ~~~
    \mathit{cov}[X_i, Y] = c_1\scov{}_{[X_i,.]}
    \label{eq:dep_one_var}
  \end{equation}
  By~\cref{thm:marginals} and~\cref{def:semantics}, we have that $\smean' = \smean''$ and $\scovprime = \scovdoubleprime$.
  Thus, $[X_1, X_2, \ldots, Y]^\intercal \sim \gauss{\smean'}{\scovprime}$ and $Y \mid X_i \sim \gauss{c_1X_i + c_2}{c_3}$.
\end{proof}

\begin{lemma}[Binary operations with values]
  Let $[X_1, X_2, \ldots]^\intercal \sim \gauss{\smean}{\scov{}}$ and $\oplus \in \{+,-,*,/\}$. 
  For all states $\s = \langle \smean, \scov{}, \sstore \rangle$, if $\langle Y = X_i \, \oplus \, e, \s \rangle \rangle \to_s \langle \smean', \scovprime, \sstore \rangle$ and $\langle \s, e \rangle \to_e c$ and $[X_1, X_2, \ldots, X_i \oplus c]^\intercal \sim \gauss{\smean''}{\scovdoubleprime}$, then\\$[X_1, X_2, \ldots, Y]^\intercal \sim \gauss{\smean'}{\scovprime}$ and $\gauss{\smean'}{\scovprime} = \gauss{\smean''}{\scovdoubleprime}$.
\end{lemma}

\begin{proof}  
  We split the proof in the cases $\oplus \in \{+,-\}$ and $\oplus \in \{*,/\}$.
  Assume w.l.o.g. that $\gauss{\smean}{\scov{}}$ contains $n$ random variables, and let $m = n+1$.

  \paragraph{Case $\oplus \in \{+,-\}$.} 
  We show the proof for $\oplus = +$; the case $\oplus = -$ is analoguous by replacing $c$ with $-c$.
  Let $\matr{A}$ be a $m \times n$ projection matrix where $\matr{A}_n = I_n$ where $I_n$ is a $n \times n$ identity matrix, $\matr{A}_{n+1}[i] = 1$ and $\matr{A}_{n+1}[j] = 0$ for $j \not= i$.
  Let $\vect{b}[i] = e$ and $\vect{b}[j] = 0$ for $j \not= i$.
  Then, by matrix multiplication $[X_1, X_2, \ldots, X_i + c]^\intercal = \matr{A}[X_1, X_2, \ldots]^\intercal + \vect{b}$.
  By~\cref{thm:affine}, $[X_1, X_2, \ldots, X_i + c]^\intercal \sim \gauss{\matr{A}\vect{\mu} + \vect{b}}{\matr{A}\scov{}\matr{A}^\intercal}$.
  By matrix multiplication
  \begin{equation}
    \matr{A}\vect{\mu} + \vect{b} = 
    \begin{bmatrix}
      \vect{\mu} \\
      \mu_{X_i} + c
    \end{bmatrix}
    ~~~
    \matr{A}\scov{}\matr{A}^\intercal = 
    \begin{bmatrix}
      \scov{} & \scov{}_{[.,X_i]} \\
      \scov{}_{[X_i,.]} & \scov{}_{[X_i,X_i]}
    \end{bmatrix}
    \label{eq:affine_sum}
  \end{equation}
  By~\cref{def:semantics} and \cref{eq:affine_sum} we have that $\smean' = \matr{A}\vect{\mu} + \vect{b}$ and $\scovprime = \matr{A}\scov{}\matr{A}^\intercal$
  Thus, $[X_1, X_2, \ldots, Y]^\intercal \sim \gauss{\smean'}{\scovprime}$ and $\gauss{\smean'}{\scovprime} = \gauss{\matr{A}\vect{\mu} + \vect{b}}{\matr{A}\scov{}\matr{A}^\intercal}$.
  This concludes the proof for $\oplus \in \{+,-\}$.

  \paragraph{Case $\oplus \in \{*,/\}$.} 
  We show the case $\oplus = *$; the case $\oplus = /$ is analoguous by replacing $c$ with $1/c$.
  Let $\matr{A}$ be a $m \times n$ projection matrix where $\matr{A}_n = I_n$ where $I_n$ is a $n \times n$ identity matrix, $\matr{A}_{n+1}[i] = c$ and $\matr{A}_{n+1}[j] = 0$ for $j \not= i$.
  Let $\vect{b} = \vect{0}$.
  Then, by matrix multiplication $[X_1, X_2, \ldots, X_ic]^\intercal = \matr{A}\vect{X} + \vect{b}$.
  By~\cref{thm:affine}, $[X_1, X_2, \ldots, X_ic]^\intercal \sim \gauss{\matr{A}\vect{\mu} + \vect{b}}{\matr{A}\scov{}\matr{A}^\intercal}$.
  By matrix multiplication
  \begin{equation}
    \matr{A}\vect{\mu} + \vect{b} = 
    \begin{bmatrix}
      \vect{\mu} \\
      \mu_{X_i}c
    \end{bmatrix}
    ~~~
    \matr{A}\scov{}\matr{A}^\intercal = 
    \begin{bmatrix}
      \scov{} & c\scov{}_{[.,X_i]} \\
      c\scov{}_{[X_i,.]} & c^2\scov{}_{[X_i,X_i]}
    \end{bmatrix}
    \label{eq:affine_prod}
  \end{equation}
  By~\cref{def:semantics} and \cref{eq:affine_prod} we have that $\smean' = \matr{A}\vect{\mu} + \vect{b}$ and $\scovprime = \matr{A}\scov{}\matr{A}^\intercal$
  Thus, $[X_1, X_2, \ldots, Y]^\intercal \sim \gauss{\smean'}{\scovprime}$ and $\gauss{\smean'}{\scovprime} = \gauss{\matr{A}\vect{\mu} + \vect{b}}{\matr{A}\scov{}\matr{A}^\intercal}$.
  This concludes the proof for $\oplus \in \{*,/\}$, and the complete of the complete lemma.
\end{proof}

\begin{lemma}[Conditioning]
  Let $[\vect{X}^\intercal_a, Y]^\intercal \sim \gauss{\smean}{\scov{}}$.
  For all states $\s = \langle \smean, \scov{}, $ $\sstore \rangle$, if $\langle \s, e \rangle \to_e c$ and $\langle \texttt{\upshape condition}(Y, c), \s \rangle \rangle \to_s \langle \smean', \scovprime, \sstore \rangle$, then $\vect{X}' \sim \gauss{\smean'}{\scovprime}$ and $\vect{X}' = \vect{X}_a \mid Y = c$.
\end{lemma}

\begin{proof}
  By~\cref{thm:conditions}, if $[\vect{X}_a^\intercal, Y]^\intercal \sim \gauss{\smean}{\scov{}}$, then $\vect{X}_a \mid Y = c \sim \gauss{\smean''}{\scovdoubleprime}$ with
  \begin{equation}
    \smean'' = \smean_a + \frac{c-\mu_Y}{\scov{}_{[Y,Y]}}\scov{}_{[.,Y]}
    ~~~
    \scovdoubleprime = \scov{}_a - \frac{1}{\scov{}_{[Y,Y]}}\scov{}_{[.,Y]}\scov{}_{[Y,.]}
    \label{eq:cond_single_RV}
  \end{equation}
  since the generalized inverse $\scov{-}_{[Y,Y]}$ of a single random variable equals $1/\scov{}_{[Y,Y]}$.
  By~\cref{def:semantics} and~\cref{eq:cond_single_RV}, we have that $\smean' = \smean''$ and $\scovprime = \scovdoubleprime$.
  Thus, $\vect{X}' \sim \gauss{\smean'}{\scovprime}$ and $\vect{X}' = \vect{X}_a \mid Y = c$.
\end{proof}

\end{document}